\documentclass[fleqn]{article}
\usepackage{graphicx}
\usepackage{amsfonts}
\usepackage{indentfirst}
\usepackage{url}
\usepackage{subfig}
\usepackage{amsmath,amssymb}
\usepackage{amsthm}
\newtheorem{theorem}{Theorem}

\newtheorem{lemma}[theorem]{Lemma}
\usepackage[utf8]{inputenc}
\usepackage[margin=1.2in]{geometry}
\newcommand{\olsi}[1]{\,\overline{\!{#1}}}

\title{\Large\textbf{Statistical Depth Function Random Variables for Univariate Distributions and induced Divergences}}
\author{Rui Ding\footnote{Corresponding author. E-mail address:rui.ding.1@stonybrook.edu }}

\begin{document}
\maketitle

\begin{abstract}
In this paper, we show that the halfspace depth random variable for samples from a univariate distribution with a notion of center is distributed as a uniform distribution on the interval $[0,\frac{1}{2}]$. The simplicial depth random variable has a distribution that first-order stochastic dominates that of the halfspace depth random variable and relates to a Beta distribution. Depth-induced divergences between two univariate distributions can be defined using divergences on the distributions for the statistical depth random variables in-between these two distributions. We discuss the properties of such induced divergences, particularly the depth-induced TVD distance based on halfspace or simplicial depth functions, and how empirical two-sample estimators benefit from such transformations.
\end{abstract}

\section{Introduction and Definitions}
Statistical depth function is a useful tool for nonparametric inference and analysis of the shape of the data, particularly for multivariate data. For a distribution $P\in\mathbb{R}^d$, a corresponding depth function is any function $D(x; P)$ which provides a P-based center-outward ordering of points $x\in\mathbb{R}^d$. Tukey \cite{Tukey75} proposed a halfspace depth function which in the univariate case is closely related to the quantile function of a distribution. The general notions of a desirable statistical depth function have been discussed by Zuo and Serfling \cite{ZS00} and Liu et al. \cite{LPS99}, among many others. Various types of statistical depth functions have been proposed, including the simplicial depth function by Liu \cite{Liu90} and depth functions based on distance functions. Zuo and Serfling \cite{ZS00} concluded that the halfspace depth function satisfies four desirable properties of a statistical depth function. In this paper, we focus mainly on the properties of halfspace depth and simplicial depth function random variables in the univariate case.

The halfspace depth function (HD) is defined for a probability measure $P$ and a point $x$ in $\mathbb{R}^d$ as:
$$HD(x; P)  =  \inf_H\{\mathbb{P}(H): x\in H\},\forall x\in \mathbb{R}^d$$
where $H$ is a closed halfspace that contains $x$ and $\mathbb{P}(\cdot)$ denotes the probability of an event. When $d=1$ and $P$ is a continuous distribution, this resolves to:
$$HD(x; P) = \min\{F(x),1-F(x)\}$$
Hence we can clearly see the relation of HD with the notion of quantile function for univariate distributions. Here $F(x)$ is defined as the cumulative distribution function(CDF) of $P$, and the probability density function(PDF) of $P$ is denoted by $f(x)$.

The simplicial depth function (SD) is defined for a probability measure $P$ and a point $x$ in $\mathbb{R}^d$ as:
$$SD(x;P) = \mathbb{P}(x\in \Delta(X_1,\ldots,X_{d+1}))$$
where $X_1,\ldots,X_{d+1}$ are i.i.d. random variables with distribution $P$ and $\Delta(X_1,\ldots,X_{d+1})$ denotes the simplex formed by these random points as vertices in $\mathbb{R}^d$. In particular, when $d=1$ this resolves to (where $\olsi{X_1 X_2}$ denotes the closed segment between $X_1,X_2\sim P$ which is the 1-D simplex):
$$SD(x;P) = \mathbb{P}(x\in \olsi{X_1 X_2})$$
which for continuous distributions $P$ can be further simplified to:
$$SD(x;P) = 2F(x)(1-F(x))$$

\section{Expectation and Distribution of the Halfspace Depth Function Random Variable}

Let us consider the random variable $X\sim P$, which has PDF $f(x)$ and CDF $F(x)$ for $x\in\mathbb{R}$. The halfspace depth function of the random variable $X$ defines another random variable which we denote by $Z = HD(X;P)$. We next study the property of this random variable $Z$. Obviously the domain of $Z = \min\{F(X),1-F(X)\}$ is $[0,\frac{1}{2}]$.

Consider first the expectation of $Z$ under the probability measure $P$. It can be calculated as:
$$\mathbb{E}_P[Z] = \int_{-\infty}^\infty \min\{F(x),1-F(x)\}f(x)dx = \int_{-\infty}^a F(x)f(x)dx+\int_{a}^{\infty} (1-F(x))f(x)dx$$
where $a$ is the notion of a center for $P$ that satisfies $F(a) = \frac{1}{2}$. The above integral has two terms, which by change of order of integration, we can observe:
$$\int_{-\infty}^a F(x)f(x)dx = \int_{-\infty}^{a}\int_{-\infty}^{x} f(y)f(x)dydx = \int_{-\infty}^{a}\int_{y}^{a} f(x)f(y)dxdy$$
$$= \int_{-\infty}^{a} (F(a) - F(y))f(y)dy = \frac{1}{2}F(a)^2 = \frac{1}{8}$$

$$\frac{1}{2}-\int_{a}^{\infty} F(x)f(x)dx = \int_{a}^{\infty} (1-F(x))f(x)dx = \int_{a}^{\infty}\int_{x}^{\infty} f(y)f(x)dydx = \int_{a}^{\infty}\int_{a}^{y} f(x)f(y)dxdy$$
$$= \int_{a}^{\infty} (F(y) - F(a))f(y)dy = \int_{a}^{\infty} F(y)f(y)dy - \frac{1}{4} = \frac{1}{8}$$
Hence we obtain $\mathbb{E}_P[HD(X;P)] = \frac{1}{4}$. This result does not require that the distribution $P$ is symmetric around its center $a$. Similarly, it can be shown that the higher-order moments are $\mathbb{E}_P[HD(X;P)^m] = \frac{2^{-m}}{m+1}$ for $m\geq1$.

We next show that the random variable $Z = HD(X;P)$ is distributed as a uniform distribution $U(0,\frac{1}{2})$ for general continuous distributions $P$ and $X\sim P$, without the requirement of $P$ being symmetric.

This can be shown be simply considering the CDF of $Z$ (where $Z$ is bounded between $[0,\frac{1}{2}]$ by the definition of the halfspace depth function):
$$F_{Z}(z) = \mathbb{P}(Z\leq z) = \mathbb{P}(F(X)\leq z)+\mathbb{P}(F(X)\geq 1-z) = 2z,\forall z\in[0,\frac{1}{2}]$$
The PDF is $f_Z(z) = 2,\forall z\in[0,\frac{1}{2}]$. Hence we conclude that $Z\sim U(0,\frac{1}{2})$. This shows that the random variable $HD(X;P)$ is a uniform distribution $U(0,\frac{1}{2})$, which has expectation $\frac{1}{4}$ as shown before, and has variance $\frac{1}{48}$.
\section{Expectation and Distribution of the Simplicial Depth Function Random Variable}

Let us consider now $X\sim P$ and the simplicial depth random variable $Z = SD(X;P)$ similar to the previous section for halfspace depth random variable, where $P$ has PDF $f(x)$ and CDF $F(x)$. Obviously the domain of $Z = 2F(X)(1-F(X))$ is $[0,\frac{1}{2}]$ similar to the HD case.

The expectation of $Z$ under the probability measure $P$ can be calculated as:
$$\mathbb{E}_P[Z] = \mathbb{E}_P[2F(X)(1-F(X))] = \int_{u=0}^{1} 2u(1-u)du = \frac{1}{3}$$
Here we used the fact that the random variable $F(X)$ is a uniform $U(0,1)$ random variable. Similarly, it can be shown that the higher-order moments are $\mathbb{E}_P[SD(X;P)^m] = 2^m\frac{\Gamma(m+1)\Gamma(m+1)}{\Gamma(2m+2)}$ for $m\geq1$.

The CDF of the random variable $Z$ can be computed as follows:
$$F_Z(z) = \mathbb{P}(Z\leq z) = \mathbb{P}(2F(X)(1-F(X))\leq z)= \mathbb{P}(F(X)\geq \sqrt{\frac{1}{4}-\frac{z}{2}}+\frac{1}{2})+\mathbb{P}(F(X)\leq \frac{1}{2}-\sqrt{\frac{1}{4}-\frac{z}{2}})$$
$$=1-\sqrt{1-2z},\forall z\in[0,\frac{1}{2}]$$
It is easy to check that $2Z\sim Beta(1,\frac{1}{2})$ is a $Beta(1,\frac{1}{2})$ random variable. The PDF of $Z$ can be obtained easily as $f_Z(z) = \frac{1}{\sqrt{1-2z}},\forall z\in[0,\frac{1}{2}]$. We can also compute the variance of $SD(X;P)$ which turns out to be $\frac{1}{45}$.

It is also easy to observe that on the interval $z\in[0,\frac{1}{2}]$, the CDF $F_{SD}(z) = 1-\sqrt{1-2z}$ is always beneath the CDF of a uniform distribution $F_{HD}(z) = 2z$ since $F_{SD}(z) = 1-\sqrt{1-2z}\leq 2z = F_{HD}(z),\forall z\in[0,\frac{1}{2}]$ with equality only taken at the two endpoints. Hence, by definition, the simplicial depth random variable $SD(X;P)$ first-order dominates the halfspace depth random variable $HD(X;P)$, denoted as $SD(X;P)\succeq_1 HD(X;P)$. The first-order stochastic dominance implies second-order $SD(X;P)\succeq_2 HD(X;P)$ and higher-order stochastic dominance relationships. 
\section{The Kernel Depth Random Variable}

A recent statistical depth function that is widely considered in machine learning applications is the kernel mean embedding, or the h-depth function \cite{WN21}. Here we define this kernel depth function as:
$$KD_k(x; P) = \mathbb{E}_{P}[k(x,X)]$$
where $X\sim P,x\in\mathbb{R}^d$, and $k(. ,.)$ is a chosen positive semi-definite kernel function, for example, the standard Gaussian kernel $k(x,y) = e^{-\frac{\|x-y\|^2}{2}},\forall x,y\in\mathbb{R}^d$. Notice that this kernel depth function also satisfies some basic desirable properties of a proper statistical depth function as defined in \cite{ZS00}. We next relate this notion of kernel depth function to an integral probability metric called the maximum mean discrepancy(MMD), defined similarly via a kernel function. Formally the MMD, also called kernel distance, is defined as:
$$MMD^2_k(P,Q) = \sup_{\|f\|_{\mathcal{H}}\leq 1} |\int fdP-\int fdQ|$$
where $k$ is the chosen positive semi-definite kernel. It has been shown that it is equivalent to:
$$MMD^2_k(P,Q) = \mathbb{E}_{X,X'\sim P}[k(X,X')]+\mathbb{E}_{Y,Y'\sim Q}[k(Y,Y')] - 2\mathbb{E}_{X\sim P,Y\sim Q}[k(X,Y)]$$
We observe that the second definition of MMD distance can be interpreted via the random variable defined through the kernel depth function. In particular, for the same choice of kernel $k$ we can write MMD as:
$$MMD^2_k(P,Q) = \mathbb{E}_{P}[KD_k(X;P)] + \mathbb{E}_{Q}[KD_k(Y;Q)] - \mathbb{E}_{P}[KD_k(X;Q)] - \mathbb{E}_{Q}[KD_k(Y;P)]$$
where $X\sim P, Y\sim Q$. Intuitively, the squared MMD distance equals the sum of two differences of kernel depth random variables with respect to their own distribution and the other distribution. 

\section{Divergence induced by Statistical Depth Function}
For a pair of probability distributions $P,Q$ where we take random variables $X\sim P,Y\sim Q$ and write their CDFs as $F_X,F_Y$, and PDFs as $f_X,f_Y$. Previous sections define statistical depth random variables based on $X,Y$ so that we can write down four random variables $HD(X;P)$, $HD(Y;P)$, $HD(X;Q)$, $HD(Y;Q)$ where the choice of statistical depth function can be replaced by simplicial depth $SD(\cdot; \cdot)$ or kernel depth $KD(\cdot; \cdot)$. Results from section 2 and 3 shows that $HD(X;P),HD(Y;Q)\sim U(0,\frac{1}{2})$ and $2SD(X;P),2SD(Y;Q)\sim Beta(1,\frac{1}{2})$. In this section, we focus on the choice of halfspace depth random variables and simplicial depth random variables. 
\subsection{Divergence induced by Halfspace Depth Function}

We can write the uniform random variable $U(0,\frac{1}{2})$ having probability distribution $U$ which has density $u(x) = 2$ everywhere on $[0,\frac{1}{2}]$. Following notations at the beginning of this section, we consider probability distributions for halfspace depth random variables $HD(X;Q)$ and $HD(Y;P)$. It can be shown that:
$$\mathbb{P}(HD(Y;P)\leq z) = \mathbb{P}(F_X(Y)\leq z) + \mathbb{P}(F_X(Y)\geq 1-z)$$
$$= \mathbb{P}(Y\leq F_X^{-1}(z))+\mathbb{P}(Y\geq F_X^{-1}(1-z)) = 1+F_Y(F_X^{-1}(z))-F_Y(F_X^{-1}(1-z))$$
$\forall z\in[0,\frac{1}{2}]$. Similarly,
$$\mathbb{P}(HD(X;Q)\leq z) =1+F_X(F_Y^{-1}(z))-F_X(F_Y^{-1}(1-z))$$
Denote the probability distributions defined by CDFs above for random variables $V=HD(Y;P)$, $W=HD(X;Q)$ as $P_Q^{HD},Q_P^{HD}$ respectively. We write $F_V(z) = 1+F_Y(F_X^{-1}(z))-F_Y(F_X^{-1}(1-z))$, $F_W(z) = 1+F_X(F_Y^{-1}(z))-F_X(F_Y^{-1}(1-z)),\forall z\in[0,\frac{1}{2}]$. It can also be shown by taking derivatives that the PDFs are given by $\forall z\in[0,\frac{1}{2}]$:
$$f_V(z) = \frac{f_Y(F_X^{-1}(z))}{f_X(F_X^{-1}(z))} + \frac{f_Y(F_X^{-1}(1-z))}{f_X(F_X^{-1}(1-z))}$$
$$f_W(z) = \frac{f_X(F_Y^{-1}(z))}{f_Y(F_Y^{-1}(z))} + \frac{f_X(F_Y^{-1}(1-z))}{f_Y(F_Y^{-1}(1-z))}$$
Consider the following divergence functions between probability distributions:
$D(P_Q^{HD}||U)$ and $D(Q_P^{HD}||U)$. Obviously when $P\neq Q$, the divergence is going to be greater than 0: $D(P_Q^{HD}||U)>0,D(Q_P^{HD}||U)>0$, and when $P=Q$, $D(P_Q^{HD}||U) = 0$. We can use the divergence between these inter-distribution statistical depth probability distributions to proxy the divergence between the original distributions $P,Q$, so this results in the following induced divergence:
$$\tilde{D}^{HD}(P||Q) = D(Q_P^{HD}||U),\tilde{D}^{HD}(Q||P) = D(P_Q^{HD}||U)$$
Here the choice of the generic divergence $D$ can be any f-divergence function, for example, the total variation distance(TVD), which is symmetric. The TVD distance between two probability distributions $P,Q$ with densities $p,q$ on the support domain $\mathcal{X}$ are defined as:
$$TVD(P,Q) = \frac{1}{2}\int_{x\in\mathcal{X}} |p(x)-q(x)|dx$$

The induced divergence $\tilde{D}^{HD}$, in general, is not symmetric hence we can symmetrize it by using:
$$\tilde{D}^{HD}_S(P,Q) = \frac{\tilde{D}^{HD}(P||Q)+\tilde{D}^{HD}(Q||P)}{2}$$
The symmetrized divergence satisfies also the property that it is zero only when $P=Q$ and otherwise greater than zero. We next show an inequality between the depth-induced TVD distance and the original TVD distance between two general distributions $P,Q$.

\begin{lemma}\label{pseudoTVDineq}
For a pair of continuous probability distributions $P,Q$, consider halfspace depth random variables with respective distributions $HD(Y;P)\sim P_Q^{HD},HD(X;Q)\sim Q_P^{HD}$, and let $U$ denote the distribution of a uniform random variable $U(0,\frac{1}{2})$. Then $TVD(P_Q^{HD},U)\leq TVD(P,Q)$ and $TVD(Q_P^{HD},U)\leq TVD(P,Q)$.
\end{lemma}
\begin{proof}
Without loss of generality, we only need to show that $TVD(Q_P^{HD},U)\leq TVD(P,Q)$. From the definitions we know $Q_P^{HD}$ has density $f_W(z) = \frac{f_X(F_Y^{-1}(z))}{f_Y(F_Y^{-1}(z))} + \frac{f_X(F_Y^{-1}(1-z))}{f_Y(F_Y^{-1}(1-z))}$ on the domain of support $z\in[0,\frac{1}{2}]$. Then,
$$TVD(Q_P^{HD},U) = \frac{1}{2}\int_{0}^{\frac{1}{2}} |f_W(z)-u(z)|dz = \frac{1}{2}\int_{0}^{\frac{1}{2}} |\frac{f_X(F_Y^{-1}(z))}{f_Y(F_Y^{-1}(z))} + \frac{f_X(F_Y^{-1}(1-z))}{f_Y(F_Y^{-1}(1-z))}-2|dz$$
$$\leq\frac{1}{2}\int_{0}^{\frac{1}{2}} (|\frac{f_X(F_Y^{-1}(z))}{f_Y(F_Y^{-1}(z))}-1| + |\frac{f_X(F_Y^{-1}(1-z))}{f_Y(F_Y^{-1}(1-z))}-1|)dz = \frac{1}{2}\int_{0}^{1} |\frac{f_X(F_Y^{-1}(z))}{f_Y(F_Y^{-1}(z))}-1|dz$$
By change of variables $y = F_Y^{-1}(z)$ which for general continuous distributions $P,Q$ and $z\in[0,1]$ have domain of support $y\in(-\infty,\infty)$ and $dy = \frac{dz}{f_Y(F_Y^{-1}(z))}$, we can rewrite inequality above as (where $f_X,f_Y$ are densities of distributions $P,Q$ respectively),
$$TVD(Q_P^{HD},U)\leq \frac{1}{2}\int_{0}^{1} |\frac{f_X(F_Y^{-1}(z))}{f_Y(F_Y^{-1}(z))}-1|dz = \frac{1}{2}\int_{-\infty}^{\infty} |f_X(y)-f_Y(y)|dy = TVD(P,Q)$$
Similarly $TVD(P_Q^{HD},U)\leq TVD(P,Q)$. This completes the proof.
\end{proof}

By Lemma \ref{pseudoTVDineq}, we have by definition $\tilde{TVD}^{HD}(P||Q)\leq TVD(P,Q)$, $\tilde{TVD}^{HD}(Q||P)\leq TVD(P,Q)$, and trivially $\tilde{TVD}^{HD}_S(P,Q)\leq TVD(P,Q)$. This shows that the (symmetrized) induced TVD based on halfspace depth random variable distributions provides a lower bound for the true TVD between original distributions. Notice that the inequality is tight under additional assumptions, such as if both distributions $P,Q$ have symmetric densities and share the same center, for example, if $P,Q$ are concentric Gaussians. In those cases we have $\tilde{TVD}^{HD}(P||Q) = \tilde{TVD}^{HD}(Q||P) = \tilde{TVD}^{HD}_S(P,Q) = TVD(P,Q)$. More generally, the conditions for the equality to hold in Lemma \ref{pseudoTVDineq} can be stated as follows. If the densities $f_X,f_Y$ of $X,Y$ satisfies,
\begin{equation}\label{eqCondition1}
(f_X(F_Y^{-1}(z)) - f_Y(F_Y^{-1}(z)))(f_X(F_Y^{-1}(1-z))-f_Y(F_Y^{-1}(1-z)))\geq 0,\forall z\in[0,\frac{1}{2}]
\end{equation} 
then $\tilde{TVD}^{HD}(P||Q) = TVD(Q_P^{HD},U) = TVD(P,Q)$. Similarly, if,
\begin{equation}\label{eqCondition2}
(f_Y(F_X^{-1}(z)) - f_X(F_Y^{-1}(z)))(f_Y(F_X^{-1}(1-z))-f_X(F_X^{-1}(1-z)))\geq 0,\forall z\in[0,\frac{1}{2}]
\end{equation} 
then $\tilde{TVD}^{HD}(Q||P) = TVD(P_Q^{HD},U) = TVD(P,Q)$. Trivially, if both \eqref{eqCondition1} and \eqref{eqCondition2} holds, then $\tilde{TVD}^{HD}_S(P,Q) = TVD(P,Q)$.

The case of symmetric and concentric distributions is a special case satisfying the conditions stated in \eqref{eqCondition1} and \eqref{eqCondition2}. We make the following definitions first. We say a distribution $P$ with support $\mathcal{X}$ is symmetric around center $x_P\in\mathcal{X}$ if $\forall x,x'\in\mathcal{X}$ such that $x+x' = 2x_P$ we have $F_X(x) = 1-F_X(x')$ and $f_X(x) = f_X(x')$. Conversely $\forall z\in[0,1]$, $F_X^{-1}(z)+F_X^{-1}(1-z) = 2x_P$ and $f_X(F_X^{-1}(z)) = f_X(F_X^{-1}(1-z))$. We say two symmetric distributions $P,Q$ having the same domain $\mathcal{X}$ and densities $f_X,f_Y$ are concentric if their centers $x_P,x_Q$ are equal: $x_P=x_Q$.
\begin{lemma}\label{concentricEq}
For two symmetric and concentric continuous distributions $P,Q$, under the same definitions with Lemma \ref{pseudoTVDineq}, we have:
$$TVD(P_Q^{HD},U) =TVD(Q_P^{HD},U) = TVD(P,Q)$$
\end{lemma}
\begin{proof}
Without loss of generality, we consider only the inequality $TVD(Q_P^{HD},U)\leq TVD(P,Q)$ in Lemma \ref{pseudoTVDineq}. This equality holds in this inequality if and only if $\forall z\in[0,\frac{1}{2}]$, if $f_X(F_Y^{-1}(z))\geq f_Y(F_Y^{-1}(z))$ then $f_X(F_Y^{-1}(1-z))\geq f_Y(F_Y^{-1}(z))$ also, and vice versa. Under assumptions that $P,Q$ are symmetric and concentric around some center $x^*$, $\forall z\in[0,\frac{1}{2}]$ we have $F_Y^{-1}(z)+F_Y^{-1}(1-z) = 2x^*$. By the same definitions: $f_X(F_Y^{-1}(z)) = f_X(F_Y^{-1}(1-z))$ and $f_Y(F_Y^{-1}(z)) = f_Y(F_Y^{-1}(1-z))$. This implies that $\frac{f_X(F_Y^{-1}(z))}{f_Y(F_Y^{-1}(z))} = \frac{f_X(F_Y^{-1}(1-z))}{f_Y(F_Y^{-1}(1-z))},\forall z\in[0,\frac{1}{2}]$. Therefore,
$$TVD(Q_P^{HD},U) = \frac{1}{2}\int_{0}^{\frac{1}{2}} |\frac{f_X(F_Y^{-1}(z))}{f_Y(F_Y^{-1}(z))}-1+\frac{f_X(F_Y^{-1}(1-z))}{f_Y(F_Y^{-1}(1-z))}-1|dz$$
$$= \frac{1}{2}\int_{0}^{\frac{1}{2}} (|\frac{f_X(F_Y^{-1}(z))}{f_Y(F_Y^{-1}(z))}-1|+|\frac{f_X(F_Y^{-1}(1-z))}{f_Y(F_Y^{-1}(1-z))}-1|)dz = \frac{1}{2}\int_{0}^{1} |\frac{f_X(F_Y^{-1}(z))}{f_Y(F_Y^{-1}(z))}-1|dz = TVD(P,Q)$$
Similarly, $TVD(P_Q^{HD},U) = TVD(P,Q)$ under the same assumptions. 
This completes the proof.
\end{proof}

Lemma \ref{concentricEq} shows that under symmetric and concentric assumptions, $$\tilde{TVD}^{HD}(P||Q) = \tilde{TVD}^{HD}(Q||P) = \tilde{TVD}^{HD}_S(P,Q) = TVD(P,Q)$$ 

We make an additional remark that the induced divergences are location and scale invariant in the sense that if random variables $X,Y$ are replaced by $aX+b,aY+b$ for some constants $a\neq0,b$, and their probability distributions denoted by $\tilde{P},\tilde{Q}$, then the distribution of random variables $HD(aX+b;\tilde{Q}),HD(aY+b;\tilde{P})$ is the same as those of $HD(X;Q),HD(Y;P)$ respectively, and $HD(aX+b;\tilde{P}),HD(aY+b;\tilde{Q})$ are still distributed as uniform distributions $U(0,\frac{1}{2})$. Hence by definition $\tilde{D}^{HD}(P||Q) = \tilde{D}^{HD}(\tilde{P}||\tilde{Q}), \tilde{D}^{HD}(Q||P) = \tilde{D}^{HD}(\tilde{Q}||\tilde{P})$. This property is desirable for all divergence functions, which is satisfied by f-divergences such as TVD. We formalize this result as Lemma \ref{invariance}.
\begin{lemma}\label{invariance}
For given constants $a\neq0,b$, and random variables $X\sim P,Y\sim Q$, denote the transformed distributions by $aX+b\sim\tilde{P},aY+b\sim\tilde{Q}$. Then $\tilde{D}^{HD}(\tilde{P}||\tilde{Q}) = \tilde{D}^{HD}(P||Q)$ for any choice of divergence function $D$ in the induced divergence.
\end{lemma}
\begin{proof}
Without loss of generality, for a given divergence $D$, we only need to show that $\tilde{D}^{HD}(\tilde{P}||\tilde{Q}) = D(\tilde{Q}^{HD}_{\tilde{P}}||U) = D(Q_P^{HD}||U)$. This amounts to showing that $HD(X;Q)$ and $HD(aX+b;\tilde{Q})$ are identically distributed random variables. We consider the CDF of $HD(aX+b;\tilde{Q})$, $\forall z\in[0,\frac{1}{2}]$:
$$\mathbb{P}(HD(aX+b;\tilde{Q})\leq z) = \mathbb{P}(F_{aY+b}(aX+b)\leq z)+\mathbb{P}(F_{aY+b}(aX+b)\geq 1-z)$$
When $a>0$, $\mathbb{P}(HD(aX+b;\tilde{Q})\leq z) = \mathbb{P}(F_{Y}(X)\leq z)+\mathbb{P}(F_{Y}(X)\geq 1-z) = \mathbb{P}(HD(X;Q)\leq z)$.
When $a<0$, $\mathbb{P}(HD(aX+b;\tilde{Q})\leq z) = \mathbb{P}(1-F_{Y}(X)\leq z)+\mathbb{P}(1-F_{Y}(X)\geq 1-z) =\mathbb{P}(F_{Y}(X)\leq z)+\mathbb{P}(F_{Y}(X)\geq 1-z)=\mathbb{P}(HD(X;Q)\leq z)$.

Combining both cases, we proved that $HD(aX+b;\tilde{Q})$ and $HD(X;Q)$ have the same distributions, which means $\tilde{Q}_{\tilde{P}}^{HD} = Q_P^{HD}$, hence the result.
\end{proof}

Trivially, the symmetrized induced divergences are also scale and location invariant: $\tilde{D}^{HD}_S(\tilde{P}||\tilde{Q}) = \tilde{D}^{HD}_S(P||Q)$. Lemma \ref{invariance} implies in particular that the (symmetrized) induced TVD ($\tilde{TVD}^{HD}$) is scale and location invariant: $\tilde{TVD}^{HD}(\tilde{P}||\tilde{Q})= \tilde{TVD}^{HD}(P||Q), \tilde{TVD}^{HD}(\tilde{Q}||\tilde{P}) = \tilde{TVD}^{HD}(Q||P)$, and $\tilde{TVD}^{HD}_S(\tilde{P},\tilde{Q}) = \tilde{TVD}^{HD}_S(P,Q)$, similar to the TVD itself which satisfies $TVD(\tilde{P},\tilde{Q}) = TVD(P,Q)$.
\subsection{Divergence under Quantile Transformation}
Closely related to the halfspace depth function transformation is the idea of transforming based on the quantile function which we denote by $QT(X;Q) = F^{-1}_{Y}(X)$, $QT(Y;P) = F^{-1}_X(Y)$, and $QT(X;P)$, $QT(Y;Q)\sim U(0,1)$ are both uniform random variables between $[0,1]$. Notice that random variables $QT(X;Q),QT(Y;P)$ have CDFs:
$$\mathbb{P}(QT(X;Q)\leq z) = F_X(F^{-1}_Y(z)),\mathbb{P}(QT(Y;P)\leq z) = F_Y(F^{-1}_X(z))$$
Following the same notation in section 5.1, we denote their probability distributions by $Q^{QT}_P,P^{QT}_Q$ respectively. Following the same definition for depth-induced divergences based on a divergence function $D$, we can define quantile-induced divergences as $\tilde{D}^{QT}(P||Q) = D(Q^{QT}_P||U),\tilde{D}^{QT}(Q||P) = D(P^{QT}_Q||U)$, where $U$ is the probability distribution of $U(0,1)$. Results in section 5.1 can be seen as a generalization of the results that also apply to the quantile transformation. We especially remark that when the divergence is TVD, we obtain the equality in Lemma \ref{pseudoTVDineq} without having to make any assumptions on the two continuous probability distributions $P,Q$.
\begin{lemma}\label{qtTVDeq}
For a pair of continuous probability distributions $P,Q$, consider quantile transformed random variables with respective distributions $QT(Y;P)\sim P_Q^{QT},QT(X;Q)\sim Q_P^{QT}$, and let $U$ denote the distribution of a uniform random variable $U(0,1)$. Then $TVD(P_Q^{QT},U)= TVD(P,Q)=TVD(Q_P^{QT},U)$. 
\end{lemma}
\begin{proof}
The proof is identical to that of Lemma \ref{pseudoTVDineq} except we don't have a triangular inequality, and the equality follows from,
\[TVD(Q^{QT}_P,U) = \frac{1}{2}\int_{0}^{1} |\frac{f_X(F^{-1}_Y(z))}{f_Y(F^{-1}_Y(z))} - 1|dz = \frac{1}{2}\int_{-\infty}^{\infty}|f_X(y)-f_Y(y)|dy = TVD(P,Q)\]
where the last step applies the change of variables $y = F^{-1}_Y(z)$ similarly. Vice versa we have $TVD(P^{QT}_Q,U) = TVD(Q,P) = TVD(P,Q)$.
\end{proof}

\begin{lemma}\label{invarianceQT}
For given constants $a\neq0,b$, and random variables $X\sim P,Y\sim Q$, denote the transformed distributions by $aX+b\sim\tilde{P},aY+b\sim\tilde{Q}$. Then $\tilde{D}_{\phi}^{QT}(\tilde{P}||\tilde{Q}) = \tilde{D}_{\phi}^{QT}(P||Q)$ for any choice of f-divergence function $D_\phi$ in the induced divergence.
\end{lemma}
\begin{proof}
Without loss of generality, for a given f-divergence $D_\phi$, we only need to show that $\tilde{D}_{\phi}^{QT}(\tilde{P}||\tilde{Q}) = D_{\phi}(\tilde{Q}^{QT}_{\tilde{P}}||U) = D_{\phi}(Q_P^{QT}||U)$. We consider the CDF of $QT(aX+b;\tilde{Q})$, $\forall z\in[0,1]$:
$$\mathbb{P}(QT(aX+b;\tilde{Q})\leq z) = \mathbb{P}(F_{aY+b}(aX+b)\leq z)$$

When $a>0$, $\mathbb{P}(QT(aX+b;\tilde{Q})\leq z) = \mathbb{P}(F_{Y}(X)\leq z) = \mathbb{P}(QT(X;Q)\leq z)$. So $\tilde{Q}^{QT}_{\tilde{P}}$ has the same probability distribution as $Q^{QT}_P$, hence the result.

When $a<0$, $\mathbb{P}(QT(aX+b;\tilde{Q})\leq z) = \mathbb{P}(F_{Y}(X)\geq 1-z) = 1-\mathbb{P}(F_{Y}(X)\leq 1-z) = 1-\mathbb{P}(QT(X;Q)\leq 1-z)$. Writing the random variables $V\sim QT(aX+b;\tilde{Q}),W\sim QT(X;Q)$ and their respective probability densities as $f_V,f_W$ where we have $f_V(z) = f_W(1-z),\forall z\in[0,1]$, we obtain (where $U$ has density $u(z)=1,\forall z\in[0,1]$):
\[D_\phi(\tilde{Q}^{QT}_{\tilde{P}}||U) = \int_{0}^{1} \phi(\frac{f_V(z)}{u(z)})u(z)dz = \int_{0}^{1} \phi(\frac{f_W(1-z)}{u(1-z)})u(1-z)d(1-z) =D_\phi(Q^{QT}_{P}||U)\]

Combining both cases, we proved that $\tilde{D}_{\phi}^{QT}(\tilde{P}||\tilde{Q}) = \tilde{D}_{\phi}^{QT}(P||Q)$ for any choice of f-divergence function $D_\phi$ in the induced divergence.
\end{proof}

TVD belongs to the family of f-divergences, and the result in Lemma \ref{invarianceQT} applies directly.

\subsection{Divergence induced by Simplicial Depth Function}

Let random variable $Z$ be that $2Z\sim Beta(1,\frac{1}{2})$ and denote the probability distribution of $Z$ by $R$ supported on the domain $z\in[0,\frac{1}{2}]$. The density of $R$ is $r(z) = f_Z(z) = \frac{1}{1-2z},\forall z\in[0,\frac{1}{2}]$ as shown in section 3.  Following notations at the beginning of this section, we consider probability distributions for simplicial depth random variables $SD(X;Q)$ and $SD(Y;P)$, where we know that $SD(X;P),SD(Y;Q)\sim R$ are both identically distributed as $Z$. It can be shown that:
$$\mathbb{P}(SD(Y;P)\leq z) = \mathbb{P}(F_X(Y)\geq \sqrt{\frac{1}{4}-\frac{z}{2}}+\frac{1}{2}) + \mathbb{P}(F_X(Y)\leq \frac{1}{2}-\sqrt{\frac{1}{4}-\frac{z}{2}})$$
$$= \mathbb{P}(Y\geq F_X^{-1}(\sqrt{\frac{1}{4}-\frac{z}{2}}+\frac{1}{2}))+\mathbb{P}(Y\leq F_X^{-1}(\frac{1}{2}-\sqrt{\frac{1}{4}-\frac{z}{2}}))$$
$$= 1-F_Y(F_X^{-1}(\sqrt{\frac{1}{4}-\frac{z}{2}}+\frac{1}{2}))+ F_Y(F_X^{-1}(\frac{1}{2}-\sqrt{\frac{1}{4}-\frac{z}{2}}))$$
$\forall z\in[0,\frac{1}{2}]$. Similarly,
$$\mathbb{P}(SD(X;Q)\leq z) =1-F_X(F_Y^{-1}(\sqrt{\frac{1}{4}-\frac{z}{2}}+\frac{1}{2}))+ F_X(F_Y^{-1}(\frac{1}{2}-\sqrt{\frac{1}{4}-\frac{z}{2}}))$$
Denote the probability distributions defined by CDFs above for random variables $V=SD(Y;P)$, $W=SD(X;Q)$ as $P_Q^{SD},Q_P^{SD}$ respectively. We write $F_V(z) = 1-F_Y(F_X^{-1}(\sqrt{\frac{1}{4}-\frac{z}{2}}+\frac{1}{2}))+ F_Y(F_X^{-1}(\frac{1}{2}-\sqrt{\frac{1}{4}-\frac{z}{2}}))$, $F_W(z) = 1-F_X(F_Y^{-1}(\sqrt{\frac{1}{4}-\frac{z}{2}}+\frac{1}{2}))+ F_X(F_Y^{-1}(\frac{1}{2}-\sqrt{\frac{1}{4}-\frac{z}{2}})),\forall z\in[0,\frac{1}{2}]$. It can also be shown by taking derivatives that the PDFs are given by $\forall z\in[0,\frac{1}{2}]$:
$$f_V(z) = \frac{1}{2\sqrt{1-2z}}(\frac{f_Y(F_X^{-1}(\frac{1}{2}-\sqrt{\frac{1}{4}-\frac{z}{2}}))}{f_X(F_X^{-1}(\frac{1}{2}-\sqrt{\frac{1}{4}-\frac{z}{2}}))} + \frac{f_Y(F_X^{-1}(\sqrt{\frac{1}{4}-\frac{z}{2}}+\frac{1}{2}))}{f_X(F_X^{-1}(\sqrt{\frac{1}{4}-\frac{z}{2}}+\frac{1}{2}))})$$
$$f_W(z) = \frac{1}{2\sqrt{1-2z}}(\frac{f_X(F_Y^{-1}(\frac{1}{2}-\sqrt{\frac{1}{4}-\frac{z}{2}}))}{f_Y(F_Y^{-1}(\frac{1}{2}-\sqrt{\frac{1}{4}-\frac{z}{2}}))} + \frac{f_X(F_Y^{-1}(\sqrt{\frac{1}{4}-\frac{z}{2}}+\frac{1}{2}))}{f_Y(F_Y^{-1}(\sqrt{\frac{1}{4}-\frac{z}{2}}+\frac{1}{2}))})$$
Consider the following divergence functions between probability distributions:
$D(P_Q^{SD}||R),D(Q_P^{SD}||R)$. Obviously when $P\neq Q$, the divergence is going to be greater than 0: $D(P_Q^{SD}||R)>0,D(Q_P^{SD}||R)>0$, and when $P=Q$, $D(P_Q^{SD}||R) = 0$. We can use the divergence between these inter-distribution statistical depth probability distributions to proxy the divergence between the original distributions $P,Q$, so this results in the following induced divergence:
$$\tilde{D}^{SD}(P||Q) = D(Q_P^{SD}||R),\tilde{D}^{SD}(Q||P) = D(P_Q^{SD}||R)$$
Similarly, we can define the symmetrized divergence as,
$$\tilde{D}^{SD}_S(P,Q) = \frac{\tilde{D}^{SD}(P||Q)+\tilde{D}^{SD}(Q||P)}{2}$$
The symmetrized divergence satisfies also the property that it is zero only when $P=Q$ and otherwise greater than zero. We next show an inequality between the simplicial depth-induced TVD distance and the original TVD distance between two general distributions $P,Q$. The properties of the simplicial depth-induced divergences are similar to those of the halfspace depth case.
\begin{lemma}\label{SDpseudoTVDineq}
For a pair of continuous probability distributions $P,Q$, consider simplicial depth random variables with respective distributions $SD(Y;P)\sim P_Q^{SD},SD(X;Q)\sim Q_P^{SD}$, and let $R$ denote a probability distribution with density $r(z) = \frac{1}{\sqrt{1-2z}},\forall z\in[0,\frac{1}{2}]$. Then $TVD(P_Q^{SD},R)\leq TVD(P,Q)$ and $TVD(Q_P^{SD},R)\leq TVD(P,Q)$.
\end{lemma}
\begin{proof}
Without loss of generality, we only need to show that $TVD(Q_P^{SD},R)\leq TVD(P,Q)$. From the definitions we know $Q_P^{SD}$ has density $f_W(z) = \frac{1}{2\sqrt{1-2z}}(\frac{f_X(F_Y^{-1}(\frac{1}{2}-\sqrt{\frac{1}{4}-\frac{z}{2}}))}{f_Y(F_Y^{-1}(\frac{1}{2}-\sqrt{\frac{1}{4}-\frac{z}{2}}))} + \frac{f_X(F_Y^{-1}(\sqrt{\frac{1}{4}-\frac{z}{2}}+\frac{1}{2}))}{f_Y(F_Y^{-1}(\sqrt{\frac{1}{4}-\frac{z}{2}}+\frac{1}{2}))})$ on the domain of support $z\in[0,\frac{1}{2}]$. Then,
$$TVD(Q_P^{SD},R) = \frac{1}{2}\int_{0}^{\frac{1}{2}} |f_W(z)-r(z)|dz$$
$$= \frac{1}{2}\int_{0}^{\frac{1}{2}}\frac{1}{2\sqrt{1-2z}}|\frac{f_X(F_Y^{-1}(\frac{1}{2}-\sqrt{\frac{1}{4}-\frac{z}{2}}))}{f_Y(F_Y^{-1}(\frac{1}{2}-\sqrt{\frac{1}{4}-\frac{z}{2}}))} + \frac{f_X(F_Y^{-1}(\sqrt{\frac{1}{4}-\frac{z}{2}}+\frac{1}{2}))}{f_Y(F_Y^{-1}(\sqrt{\frac{1}{4}-\frac{z}{2}}+\frac{1}{2}))}-2|dz$$
Here we make the change of variable $y = \sqrt{\frac{1}{4}-\frac{z}{2}}$ and $dy = \frac{-dz}{\sqrt{1-2z}}$. Then we have,
$$TVD(Q_P^{SD},R)\leq\frac{1}{2}\int_{0}^{\frac{1}{2}} (|\frac{f_X(F_Y^{-1}(\frac{1}{2}-y))}{f_Y(F_Y^{-1}(\frac{1}{2}-y))}-1| + |\frac{f_X(F_Y^{-1}(\frac{1}{2}+y))}{f_Y(F_Y^{-1}(\frac{1}{2}+y))}-1|)dy = \frac{1}{2}\int_{0}^{1} |\frac{f_X(F_Y^{-1}(y))}{f_Y(F_Y^{-1}(y))}-1|dy$$
By another change of variables $x = F_Y^{-1}(y)$ which for general continuous distributions $P,Q$ and $y\in[0,1]$ have domain of support $x\in(-\infty,\infty)$ and $dx = \frac{dy}{f_Y(F_Y^{-1}(y))}$, we can rewrite inequality above as (where $f_X,f_Y$ are densities of distributions $P,Q$ respectively),
$$TVD(Q_P^{SD},R)\leq \frac{1}{2}\int_{0}^{1} |\frac{f_X(F_Y^{-1}(y))}{f_Y(F_Y^{-1}(y))}-1|dy = \frac{1}{2}\int_{-\infty}^{\infty} |f_X(x)-f_Y(x)|dx = TVD(P,Q)$$
Similarly $TVD(P_Q^{SD},R)\leq TVD(P,Q)$. This completes the proof.
\end{proof}

Similar to Lemma \ref{concentricEq}, we can establish the equality:
$$\tilde{TVD}^{SD}(P||Q) = \tilde{TVD}^{SD}(Q||P) = \tilde{TVD}_S^{SD}(P,Q) = TVD(P,Q)$$ 
in the inequalities from Lemma \ref{SDpseudoTVDineq} under specific assumptions \eqref{eqCondition1} and \eqref{eqCondition2}, including the special case when $P,Q$ are symmetric and concentric. Similar to Lemma \ref{invariance}, we can also establish the location and scale invariance of $\tilde{D}^{SD}(P||Q),\tilde{D}^{SD}(Q||P)$,and trivially $\tilde{D}^{SD}_S(P,Q)$. These are established in the following lemmas where the proof can be simply derived from that of Lemma \ref{concentricEq} and \ref{invariance}.
\begin{lemma}\label{SDconcentricEq}
For two symmetric and concentric continuous distributions $P,Q$, under the same definitions with Lemma \ref{SDpseudoTVDineq}, we have:
$$TVD(P_Q^{SD},R) =TVD(Q_P^{SD},R) = TVD(P,Q)$$
\end{lemma}

\begin{lemma}\label{SDinvariance}
For given constants $a\neq0,b$, and random variables $X\sim P,Y\sim Q$, denote the transformed distributions by $aX+b\sim\tilde{P},aY+b\sim\tilde{Q}$. Then $\tilde{D}^{SD}(\tilde{P}||\tilde{Q}) = \tilde{D}^{SD}(P||Q)$ for any choice of divergence function $D$ in the induced divergence.
\end{lemma}

\section{Numerical Results}
We demonstrate the behavior of the halfspace depth random variables and simplicial depth random variables in a simple numerical experiment and demonstrate how induced TVD can be estimated from the statistical depth distributions and provide a close estimate of the true TVD between two data distributions. An empirical Lipschitz variational TVD (LV-TVD) estimator is used to estimate a variational lower bound of the true TVD between two distributions based on data samples from them respectively, see \cite{DingICSTA23}. The estimator can be applied to estimating either $TVD(P,Q)$ from the original data samples, or $\tilde{TVD}^{HD}(P||Q),\tilde{TVD}^{HD}(Q||P)$ and $\tilde{TVD}^{SD}(P||Q),\tilde{TVD}^{SD}(Q||P)$ from depth function transformed samples, where the domain is bounded in $[0,\frac{1}{2}]$. It can also be applied to $\tilde{TVD}^{QT}(P||Q),\tilde{TVD}^{QT}(Q||P)$ from quantile transformed samples, where the domain is bounded in $[0,1]$, which should have an exact same target value as the ground-truth $TVD(P,Q)$.

Consider two symmetric and concentric Gaussian distributions $P = \mathcal{N}(0,1),Q = \mathcal{N}(0,1.5^2)$. Here the ground-truth TVD between $P,Q$ is $TVD(P,Q) = 0.19358$. Based on Lemma \ref{concentricEq} and \ref{SDconcentricEq}, we know that $\tilde{TVD}^{HD}(Q||P) = \tilde{TVD}^{HD}(P||Q) = \tilde{TVD}^{HD}_S(P,Q) = TVD(P,Q) =\tilde{TVD}^{SD}(Q||P) = \tilde{TVD}^{SD}(P||Q) = \tilde{TVD}^{SD}_S(P,Q)$ for this case. 
\begin{figure}[!ht]
    \centering
    \includegraphics[scale=0.5]{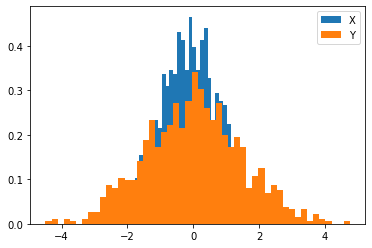}
    \caption{Data Sample $X$ from $P$ and $Y$ from $Q$}
    \label{fig:demo_data}
\end{figure}

We randomly sampled $N=1000$ points from $P$ and $Q$, denoted by data samples $\{X_i\}_{i=1}^N,\{Y_j\}_{j=1}^N$ which respectively define empirical distributions $P_N,Q_N$ that converges to $P,Q$ when $N\to\infty$. In Figure \ref{fig:demo_data}, the data samples are plotted for the two distributions. The transformed data samples $\{HD(Y_j;P_N)\}_{j=1}^N$ are empirical observations from the halfspace depth random variables $HD(Y;P)$, and similarly $\{HD(X_i;P_N)\}_{i=1}^N$ are empirical observations from $HD(X;P)$ which is known to be a uniform random variable $U(0,\frac{1}{2})$. For empirical distribution $P_N$, the halfspace depth observations are computed as $HD(x;P_N) = \frac{\min\{\sum_{i=1}^N \textbf{1}\{X_i\leq x\},\sum_{i=1}^N \textbf{1}\{X_i\geq x\},0.5\}}{N},\forall x\in\mathbb{R}$, which is always bounded between $[0,\frac{1}{2}]$. 

\begin{figure}[!ht]
    \centering
    \includegraphics[scale=0.5]{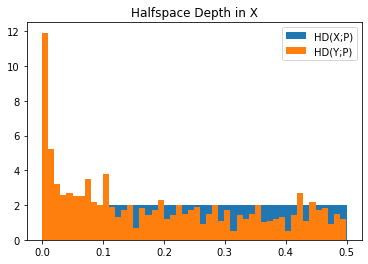}
    \caption{Halfspace Depth Data Samples for $X$ and $Y$ against Empirical Distribution of $P$}
    \label{fig:hdrvs}
\end{figure}

We plot in Figure \ref{fig:hdrvs} the empirical observations of halfspace depth random variables $HD(X;P)$ and $HD(Y;P)$, which shows that the samples $\{HD(X_i;P_N)\}_{i=1}^N$ ,corresponding to the random variable $HD(X;P)$, obviously follow a uniform distribution $U(0,\frac{1}{2})$. Notice that these random observations, for a fixed data size $N$ of $P_N$, corresponds to a discrete support that is shared by the two samples; therefore one would like to think that an empirical TVD computed from the discrete distribution based on the empirical samples can provide an estimate for the ground-truth TVD. However, this usually results in an overestimate. Hence, the LV-TVD estimator is preferred in this task. (A discrete TVD computation would correspond conceptually to $l=\infty$ in the LV-TVD estimator.)

Computing the empirical LV-TVD between $P_N,Q_N$ with Lipschitz parameter $l=4$ gives an estimate of $TVD(P,Q)$ around $0.19497$. The empirical LV-TVD computed using halfspace depth samples $\{HD(X_i;P_N)\}_{i=1}^N,\{HD(Y_j;P_N)\}_{j=1}^N$ gives an estimate of $TVD(P_Q^{HD},U) = \tilde{TVD}^{HD}(Q||P)$ which is $0.19220$ (using $l=20$). Vice versa we can estimate $TVD(Q_P^{HD},U) = \tilde{TVD}^{HD}(P||Q)$ with a similar procedure using samples $\{HD(X_i;Q_N)\}_{i=1}^N,\{HD(Y_j;Q_N)\}_{j=1}^N$ and the empirical LV-TVD estimate is $0.19028$. The estimate for symmetrized induced TVD $\tilde{TVD}^{HD}_S(P,Q)$ is hence $0.19124$. The estimate based on original data samples and based on transformed halfspace depth samples are fairly close, and both close to the real TVD value.

Similarly, Figure \ref{fig:sdrvs} shows the empirical observations of simplicial depth random variables $SD(X;P)$, $SD(Y;P)$, which shows that the samples $\{SD(X_i;P_N)\}_{i=1}^N$, corresponding to the random variable $SD(X;P)$, obviously follow a Beta distribution (on the half unit interval) $\frac{Beta(1,\frac{1}{2})}{2}$. 

For empirical distribution $P_N$, the simplicial depth observations are computed as $SD(x;P_N) = 2\frac{\sum_{i=1}^N \textbf{1}\{X_i\leq x\}}{N}\frac{\sum_{i=1}^N \textbf{1}\{X_i> x\}}{N}$, $\forall x\in\mathbb{R}$, which is always bounded between $[0,\frac{1}{2}]$.
\begin{figure}[!ht]
    \centering
    \includegraphics[scale=0.5]{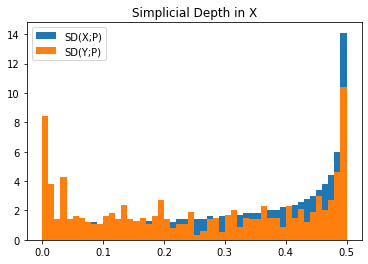}
    \caption{Simplicial Depth Data Samples for $X$ and $Y$ against Empirical Distribution of $P$}
    \label{fig:sdrvs}
\end{figure}

The LV-TVD estimate using the samples $\{SD(X_i;P_N)\}_{i=1}^N,\{SD(Y_j;P_N)\}_{j=1}^N$ (with $l=20$ also) gives 0.19562. Vice versa, the LV-TVD estimate using samples $\{SD(X_i;Q_N)\}_{i=1}^N,\{SD(Y_j;Q_N)\}_{j=1}^N$ (with $l=20$ also) gives 0.19215, which are also close to the estimate of TVD based on the original data samples as expected. The estimate for symmetrized induced TVD $\tilde{TVD}^{SD}_S(P,Q)$ is hence 0.193885.

We ignored the quantile transformed samples here since they behave similarly to halfspace depth samples but on a larger domain. The estimation procedure is similar and empirical samples of $QT(x;P_N)$ can be calculated as $\frac{\sum_{i=1}^N \textbf{1}\{X_i\leq x\}}{N}$, which is exactly a uniform grid of $\{\frac{1}{N},\ldots,\frac{N}{N}=1\}$ in $[0,1]$ if samples points $\{X_i\}_{i=1}^N$ are considered against its own empirical distribution $P_N$. The LV-TVD estimated value ($l=20$) using the quantile transformed samples give a very stable result of $0.19452$ for both $\tilde{TVD}^{QT}(P||Q)$ and $\tilde{TVD}^{QT}(Q||P)$ which should be the same value as the ground-truth $TVD(P,Q)$. These values are higher than HD or SD estimated ones but closer to the LV-TVD estimate between the original data samples since the quantile transformation step preserves the ground-truth TVD value between $P,Q$ and hence not a lower bound approximation. 
\section{Performance Improvement in LV-TVD Estimation Procedures for Halfspace Depth-induced TVD}
In the numerical example above, we considered estimating an LV-TVD value for $TVD(P_Q^{HD},U)$ based on two halfspace depth random variable samples $\{HD(X_i;P_N)\}_{i=1}^N$ and $\{HD(Y_j;P_N)\}_{j=1}^N$. Vice versa, for estimating $TVD(Q_P^{HD},U)$, we used two halfspace depth random variable samples $\{HD(X_i;Q_N)\}_{i=1}^N$ and $\{HD(Y_j;Q_N)\}_{j=1}^N$. Since we know the random samples $\{HD(X_i;P_N)\}_{i=1}^N$ and $\{HD(Y_j;Q_N)\}_{i=1}^N$ are distributed as the uniform distribution $U(0,\frac{1}{2})$, but we are not using this information in the original LV-TVD estimation procedure, hence introducing additional variance into the estimated values. Next, we present two techniques to handle the uniform distribution on the right-hand side of these induced TVDs. The same techniques can be applied to quantile transformed samples when estimating $\tilde{TVD}^{QT}(P||Q),\tilde{TVD}^{QT}(Q||P)$ based on the uniform distribution $U(0,1)$.

\subsection{Improving Estimator Performance using Finer Samples from $U(0,\frac{1}{2})$}
As an obvious extension, we can use coarser samples from $U(0,\frac{1}{2})$ to replace halfspace depth samples $\{HD(X_i;P_N)\}_{i=1}^N$ or  $\{HD(Y_j;Q_N)\}_{j=1}^N$, which we know are distributed according to $U(0,\frac{1}{2})$. Notice that these samples are in fact $\frac{N}{2}$ uniformly spaced samples $\{\frac{1}{N},\frac{2}{N},\ldots,\frac{1}{2}\}$ for a given even number of samples $N$, where each value occurs exactly twice. (Similar results can be obtained for an odd number of samples $N$ with a slight difference.) We can take a much larger even number of samples $\{U_i\}_{i=1}^{M}$ which corresponds to taking each of $\{\frac{1}{M},\frac{2}{M},\ldots,\frac{1}{2}\}$ exactly twice. Using this sample in replacement of $\{HD(X_i;P_N)\}_{i=1}^N$ or  $\{HD(Y_j;Q_N)\}_{j=1}^N$ in the LV-TVD estimator for $TVD(P_Q^{HD},U)$ or $TVD(Q_P^{HD},U)$, we can improve the convergence behavior of these estimators and get estimates with less variance. The trade-off is that with more samples, the LP problem in the LV-TVD estimator contains more decision variables. Consider the same data sample as in Figure \ref{fig:demo_data}. Numerical results in section 6 report that using the two-sample LV-TVD procedure for halfspace depth data samples $(l=20)$, the estimated values are $0.19220$ and $0.19028$ respectively (hence a symmetrized estimate of $0.19124$), whereas the ground-truth value is $TVD(P_Q^{HD},U) = TVD(Q_P^{HD},U) = TVD(P,Q) = 0.19358$ and the straightforward LV-TVD of original data samples $(l=4)$ gives $0.19497$. 

Applying a more refined uniform sample $\{U_i\}_{i=1}^{M}$ with $M = 2N = 2000$ in replacement of the empirical samples $\{HD(X_i;P_N)\}_{i=1}^N$ or  $\{HD(Y_j;Q_N)\}_{j=1}^N$, we obtain LV-TVD estimates (using $l=20$) of $0.19115$ and $0.19214$ for $TVD(P_Q^{HD},U)$ and $TVD(Q_P^{HD},U)$ respectively, and hence a symmetrized estimate of $\tilde{TVD}^{HD}_S(P,Q) = 0.19165$. Similarly, for $\{U_i\}_{i=1}^{M}$ with $M = 4N = 4000$, we obtain LV-TVD estimates (using $l=20$) of $0.19154$ and $0.19202$ for $TVD(P_Q^{HD},U)$ and $TVD(Q_P^{HD},U)$ respectively, and hence a symmetrized estimate of $\tilde{TVD}^{HD}_S(P,Q) = 0.19178$. Clearly, the proposed extension with finer and finer samples from $U(0,\frac{1}{2})$ increasingly reduces the variance as well as improves the convergence behavior of the LV-TVD estimators toward their target TVD values. As a remark, a similar technique can be derived for the simplicial depth random variables by using more refined samples from the Beta distribution with similar spacing as empirical samples $\{SD(X_i;P_N)\}_{i=1}^N$ and $\{SD(Y_j;Q_N)\}_{j=1}^N$. For quantile transformed samples, it is mostly similar to the halfspace depth case, where we can use a finer grid of $M$ uniform samples $\{\frac{1}{M},\frac{2}{M},\ldots,1\}$ from the standard uniform distribution $U(0,1)$. For $\tilde{TVD}^{QT}(P||Q)=\tilde{TVD}^{QT}(Q||P)$, the LV-TVD estimated results ($l=20$) based on $M=2N=2000$ samples are both $0.19498$, using their quantile transformed samples respectively.

\subsection{Variance Reduction with One-sided Estimators and Restricted Function Class}
We next propose another modification of the LV-TVD estimator that directly takes in one sample $\{Z_i\}_{i=1}^N$ in 1-D and estimates its LV-TVD against a ground-truth uniform distribution $U(a,b),b>a$, where the domain of empirical samples $\{Z_i\}_{i=1}^N$ is also $[a,b]$. See \cite{DingICSTA23} for more details. Without loss of generality let $\{Z_i\}_{i=1}^N$ be given in increasing order. This one-sided LV-TVD estimator relies on the extension of optimal identifier functions $\{f^\star(Z_i)\}_{i=1}^N$ to the entire domain of $[a,b]$ based on a piecewise linear interpolation, and a constant extension at the two endpoints. Following the notations in \cite{DingICSTA23}, the original LV-TVD distance between empirical distribution $P^Z_{N}$ of $\{Z_i\}_{i=1}^N$ and uniform distribution $U(a,b)$ is:
$$\gamma^{l}_{LVD}(P^Z_N,U) = \frac{1}{2} \sup_{f\in\{f:||f||_{L}\leq l,||f||_\infty\leq 1\}}\{\frac{1}{N}\sum_{i=1}^N f(Z_i)-\frac{1}{b-a}\int_{a}^b f(x)dx\}$$
In the LP formulation of the above empirical distance, the decision variables are $a_i = f(Z_i),\forall i=1,\ldots,N$, and the integral can be rewritten based on the piecewise linear interpolation of $\{f(Z_i)\}_{i=1}^N$ to the entire domain $[a,b]$. This results in the following LP problem, where $\{Z_i\}_{i=1}^N\in[a,b]$ are assumed to be in non-decreasing order:

\begin{equation}\label{empLVLP_uniform}
\begin{aligned}
    \tilde{\gamma}^l_{LVD}(P^Z_N,U) &= \frac{1}{2} \max_{a_1,...,a_N}\{\frac{1}{N}\sum_{i=1}^N a_i - \frac{(Z_1 - a) a_1}{b-a} - \sum_{i=1}^{N-1} \frac{(Z_{i+1}-Z_i)(a_{i+1}+a_{i})}{2(b-a)} - \frac{(b-Z_{N})a_{N}}{b-a}\} \\
    s.t. &-l(Z_{i+1}-Z_i)\leq a_{i+1}-a_i\leq l(Z_{i+1}-Z_i),\forall i=1,...,N-1\\
    &-1\leq a_i\leq1,\forall i=1,...,N
\end{aligned}
\end{equation}
The problem in \eqref{empLVLP_uniform} follows from a reduction of the constraints in the original LP in 1-D setting, and we have $\tilde{\gamma}^{l}_{LVD}(P^Z_N,U)\leq \gamma^{l}_{LVD}(P^Z_N,U)$, see discussions in \cite{DingICSTA23}. 

The optimal objective value of this LP gives an asymptotic lower bound of LV-TVD estimate of the induced TVD distance $\tilde{TVD}^{HD}(Q||P) = TVD(P^{HD}_Q,U)$, when the input data samples $\{Z_j\}_{j=1}^N$ are statistical depth random variables $\{HD(Y_j;P_N)\}_{j=1}^N$ and the domain $[a,b] = [0,\frac{1}{2}]$. This is a lower bound because we restricted the identifier functions based on piecewise linear interpolations of the node values. Generally, this lower bound is relatively tight. Similarly, using this one-sided LV-TVD estimator \eqref{empLVLP_uniform} we can estimate a tight lower bound of $\tilde{TVD}^{HD}(P||Q) = TVD(Q^{HD}_P,U)$ when the input data samples $\{Z_i\}_{i=1}^N$ are statistical depth random variables $\{HD(X_i;Q_N)\}_{i=1}^N$. Trivially, we also obtain an estimate of the symmetrized induced TVD $\tilde{TVD}^{HD}_S(P,Q)$ based on these two one-sided LV-TVD estimates. To demonstrate the behavior of one-sided LV-TVD estimators based on ground-truth uniform distributions, we applied it to the problem in section 6 for the HD case. For the same data samples shown in Figure \ref{fig:demo_data}, consider the induced TVD $TVD(P_Q^{HD},U)$ and $TVD(Q_P^{HD},U)$. Numerical results in section 6 report that using the two-sample LV-TVD procedure for halfspace depth data samples $(l=20)$, the estimated values are $0.19220$ and $0.19028$ respectively, whereas the ground-truth value is $TVD(P_Q^{HD},U) = TVD(Q_P^{HD},U) = TVD(P,Q) = 0.19358$ and the straightforward LV-TVD of original data samples $(l=4)$ gives $0.19497$.

Using the proposed one-sided restricted LV-TVD estimator $\tilde{\gamma}^{l}_{LVD}(P^Z_N,U)$ for $TVD(P_Q^{HD},U)$ and $TVD(Q_P^{HD},U)$, where $U$ is a given uniform distribution $U(0,\frac{1}{2})$ and $l=20$ as before, we obtain estimated values $0.19071,0.19075$ respectively, which are smaller than the two-sample LV-TVD estimated values but show significantly smaller variance as well. This effect should be more significant, especially when the data size is smaller, where the ground-truth uniform distribution is represented by coarser empirical samples. Hence, the estimator in \eqref{empLVLP_uniform} serves as a stabilizing technique for this type of estimation problems based on halfspace depth random variables. Similar technique should also apply to the simplicial depth case, although the ground-truth density of the Beta distribution is harder to work with in the LP problem for the empirical LV-TVD distance.

\subsection{Summary of Results}
Table \ref{tab1} summarizes the two techniques we discussed in this section and their improved performance in terms of variance reduction and/or improved convergence based on the numerical example we provided in section 6, as compared against directly using halfspace depth samples for both distributions. Again, the ground-truth value is $TVD(P,Q) = 0.19358$ and the direct LV-TVD estimate ($l=4$) based on original data samples is $0.19497$. All LV-TVD estimators for induced TVD estimates based on halfspace depth observations use $l=20$ as the Lipschitz parameter. The original data sample size is $N = 1000$ for both distributions. We label the different techniques in Table \ref{tab1} based on different approaches to handle the ground-truth uniform distribution $U(0,\frac{1}{2})$ in all the halfspace depth-induced TVD estimates.

\begin{table}[!ht]
    \centering
      \caption{\centering Empirical LV-TVD Estimators for induced TVD based on Halfspace Depth Samples}
      \scalebox{0.85}{
       \begin{tabular}{|| c | c | c | c ||} 
 \hline
Input Samples/Techniques for $U(0,\frac{1}{2})$ & $\tilde{TVD}^{HD}(Q||P)$ & $\tilde{TVD}^{HD}(P||Q)$ & $\tilde{TVD}^{HD}_{S}(P,Q)$ \\  
 \hline
Empirical samples of size $N$ (section 6) & 0.19028 & 0.19220 & 0.19124\\
\hline
More refined samples of size $M=2N$ (section 7.1) & 0.19115 & 0.19214 & 0.19165\\
\hline
More refined samples of size $M=4N$ (section 7.1) & 0.19154 & 0.19202 & 0.19178\\
\hline
Density-based variational lower bound (section 7.2) & 0.19071 & 0.19075 & 0.19073\\
\hline
\end{tabular}}
\label{tab1}
\end{table}

\section{Conclusion}

We show that the halfspace depth random variable $HD(X;P)$, where $X\sim P$ from a univariate continuous probability distribution, is distributed as a uniform distribution $U(0,\frac{1}{2})$, regardless of $P$ being symmetric or not. The distribution function for the simplicial depth random variable is also computed, which turns out to be first-order stochastic dominant over that of the halfspace depth random variable, having a larger mean and smaller variance. We also discussed the kernel depth function and its relation with maximum mean discrepancy. Finally, we propose a depth-induced divergence for two distributions based on divergences between statistical depth distributions in-between them and showed specific cases using halfspace or simplicial depth functions and with total variation distance as the divergence function. In particular, we show how an empirical Lipschitz variational total variation distance estimator benefits from such transformations.

\end{document}